\documentclass[letterpaper, 10 pt, conference]{ieeeconf} 
\IEEEoverridecommandlockouts         
\usepackage{graphicx}  
\usepackage{xargs, xifthen} 
\usepackage{soul, color, xcolor}
\usepackage{amsmath}

\usepackage{amsthm}
\usepackage{tikz}
\usetikzlibrary{positioning,arrows,calc,automata}
\tikzset{>=latex}
\usepackage{booktabs} 
\usepackage{subcaption}
\DeclareMathAlphabet{\pazocal}{OMS}{zplm}{m}{n}
\usepackage[ruled,vlined,linesnumbered]{algorithm2e}
\usepackage{hyperref}
\usepackage{tabularx} 
\usepackage{xspace} 

\newcommand{\eq}{\begin{equation}}
\newcommand{\qe}{\end{equation}}
\newcommand{\eqs}{\begin{equation*}}
\newcommand{\qes}{\end{equation*}}
\newtheorem{example}{Example}
\newcommand{\posint}{\mathbf{N}}
\newcommand{\natnum}{\mathbf{N}_0}

\newcommand{\real}{\mathbf{R}}

\newcommand{\ceil}[1]{\left\lceil#1\right\rceil}

\newcommand{\eqdist}{\stackrel{d}{=}}
\newcommandx{\gr}[1][1=]{\mathcal{G}%
	\ifthenelse{\isempty{#1}}{}{(#1)}%
}
\newcommand{\nodes}{\mathcal{N}}
\newcommandx{\arcs}[1][1=]{\mathcal{A}%
	\ifthenelse{\isempty{#1}}{}{(#1)}%
}
\newcommandx{\neigh}[2][1=k,2=i]{\mathcal{N}_{#2}%
	\ifthenelse{\isempty{#1}}{}{(#1)}%
}
\newcommandx{\neighinf}[2][1=k,2=i]{\bar{\mathcal{N}}_{#2}%
	\ifthenelse{\isempty{#1}}{}{(#1)}%
}
\newcommand{\tdisc}[1][k]{#1\in\natnum}
\newcommand{\aginset}[1][i]{#1\in\nodes}
\newcommandx{\coeff}[3][1=i,2=j,3=k]{\xi_{#1 #2}
	\ifthenelse{\isempty{#3}}{}{{(#3)}}%
}
\newcommand{\coeffmatrix}{\Xi}
\newcommand{\Hea}[1][i]{H_i}
\newcommandx{\noise}[2][1=i,2=k]{\eta_{#1}(#2)}
\newcommandx{\extinf}[2][1=i,2=k]{u_{#1}(#2)}
\newcommand{\probab}[1]{\mathrm{P}\left(#1\right)}

\newcommandx{\x}[2][1=i,2=k]{x_{#1}%
	\ifthenelse{\isempty{#2}}{}{(#2)}%
}
\newcommandx{\s}[2][1=i,2=k]{s_{#1}%
	\ifthenelse{\isempty{#2}}{}{(#2)}%
}
\newcommandx{\xvec}[1][1=k]{\mathbf{x}%
	\ifthenelse{\isempty{#1}}{}{(#1)}%
}
\newcommandx{\svec}[1][1=k]{\mathbf{s}%
	\ifthenelse{\isempty{#1}}{}{(#1)}%
}
\newcommandx{\y}[2][1=i,2=k]{y_{#1}%
	\ifthenelse{\isempty{#2}}{}{(#2)}%
}
\newcommandx{\C}[2][1=i,2=k]{\chi_{#1}%
	\ifthenelse{\isempty{#2}}{}{(#2)}%
}
\newcommandx{\ki}[1][1=i]{k_{#1}
}

\newcommandx{\pr}[2][1=i,2=k]{p_{#1}(#2)}
\newcommand{\exnodesnotlab}{\nodes^{\mathrm{ext}}}
\newcommand{\exnodes}{\tilde{\nodes}}
\newcommandx{\exi}[1][1=i]{-{#1}}
\newcommandx{\tforinf}[2][1=i,2=j]{\tau_{{#1}{#2}}}
\newcommandx{\ttilde}[2][1=i,2=j]{\tilde{\tau}_{{#1}{#2}}}
\newcommandx{\estforinf}[2][1=i,2=j]{\hat{\tau}_{{#1}{#2}}}
\newcommandx{\inftime}[1][1=i]{\ki[{#1}]}
\newcommandx{\estinftime}[1][1=i]{\hat{t}_{#1}}
\newcommandx{\tforrec}[1][1=i]{r_{#1}}
\newcommandx{\tforex}[1][1=i]{k_{#1}^{\mathrm{ext}}}

\newcommandx{\exinput}[2][1=i,2=k]{x_{\exi[#1]}(#2)}
\newcommandx{\rectime}[1][1=i]{R_{#1}}
\newcommandx{\numinf}[1][1=k]{\lambda%
    \ifthenelse{\isempty{#1}}{}{(#1)}
}
\newcommandx{\thres}[1][1=i]{\phi_{#1}}
\newcommand{\incircbin}[1]{%
  \mathbin{%
    \mathchoice%
    {\protect\incircint{\displaystyle}{#1}}%
    {\protect\incircint{\textstyle}{#1}}%
    {\protect\incircint{\scriptstyle}{#1}}%
    {\protect\incircint{\scriptscriptstyle}{#1}}%
  }%
}
\newcommand{\incircint}[2]{%
  \ooalign{$#1\bigcirc$\crcr\hidewidth$#1#2$\hidewidth\crcr}%
}
\newcommand{\circleland}{\incircbin{\land}}

\newcommandx{\xbo}[2][1=k, 2=]{\mathbf{x}^b%
	\ifthenelse{\isempty{#2}}{}{_{#2}}%
	(#1)
}
\newcommandx{\ybo}[2][1=k, 2=]{\mathbf{y}^b%
	\ifthenelse{\isempty{#2}}{}{_{#2}}%
	(#1)
}

\newtheorem{assumption}{Assumption}
\newtheorem{lemma}{Lemma}
\newtheorem{definition}{Definition}
\newtheorem{proposition}{Proposition}
\newtheorem{corollary}{Corollary}
\newtheorem{theorem}{Theorem}
\newtheorem{remark}{Remark}

\title{Low Complexity Method for Simulation of Epidemics\\Based on Dijkstra's Algorithm} 

\author{Davide Zorzenon, Fabio Molinari, J\"org Raisch
	\thanks{D. Zorzenon, F. Molinari, and J. Raisch
	are with
	the Control Systems Group within the Technische Universit\"at Berlin, Germany.
	\emph{\{zorzenon,molinari,raisch\}@control.tu-berlin.de}.}
}
\begin{document}
	\maketitle
	\thispagestyle{empty}
	\pagestyle{empty}	
	
	\begin{abstract}   
Models of epidemics over networks have become
popular, as they describe the impact of
individual
behavior on infection spread.
However, they come with high computational complexity,
which constitutes a problem in case large-scale scenarios
are considered.
This paper presents a 
discrete-time multi-agent SIR (Susceptible, Infected, Recovered)
model that extends known results in literature.
Based on that, using the novel notion of \textit{Contagion Graph}, it proposes a graph-based method
derived from Dijkstra's algorithm
that allows to decrease the computational complexity of a simulation. 
The Contagion Graph can be also employed as an approximation scheme
describing the ``mean behavior'' of an epidemic
over a network and requiring low computational
power. Theoretical findings are confirmed by
randomized large-scale simulation.
%

	\end{abstract}
	
	\section{Introduction}

Mathematical models of epidemics are essential tools for forecasting the spread of diseases.
Recently, the \mbox{COVID-19} pandemic has shown the importance of such instruments for planning control measures and assessing their impact~\cite{giordano2020sidarthe,bouchnita2020hybrid}.
Among the frameworks available in literature, network-based models are particularly suitable for evaluating non-pharmaceutical interventions, e.g., social distancing, since they provide a natural representation of contact interactions between individuals (see, e.g., Fig.~\ref{fig:network})~\cite{teweldemedhin2004agent,youssef2013mitigation}.
However, computational complexity is an issue when large-scale scenarios are considered~\cite{duan2015mathematical}.
Consequently, researchers have suggested several strategies aiming to reduce the computational burden while maintaining adequate levels of accuracy.

For instance, \cite{koher2016infections} introduces a dynamical discrete-time Susceptible, Infected, Recovered (SIR) model with Boolean algebra formalism, in which operations can be efficiently implemented.
SIR models, in which individuals are assumed to get immune after recovering from the disease, have been widely used to describe infectious diseases~\cite{calafiore2020modified,osthus2017forecasting}.
However, the limitation of discrete-time simulations comes from the synchronous updating, in which the state of the system is updated at regular time intervals.
Since the time step has to be small enough to capture high frequency phenomena, the state is often refreshed even when no new events occur, causing a waste of computational resources.
Event-based methods, as the Gillespie algorithm, overcome this problem by updating the state only when a new event occurs~\cite{vestergaard2015temporal,ahmad2019continuous}.

In our paper, we first present a discrete-time SIR model of the disease evolution.
Our model is a dynamical and stochastic generalization of the Boolean model described in~\cite{koher2016infections}, 
which also captures the possibility that individuals can be infected from outside the investigated network at an arbitrary time.
Next, we introduce the so-called \emph{Contagion Graph}, a graph obtained from the interaction network, from which we formally derive an event-driven procedure to simulate the model.
We analytically show that the procedure, based on Dijkstra's algorithm, significantly reduces the computational
complexity, thus simulation time.
\begin{figure}[t]
	\centering
	\includegraphics[width=\columnwidth]{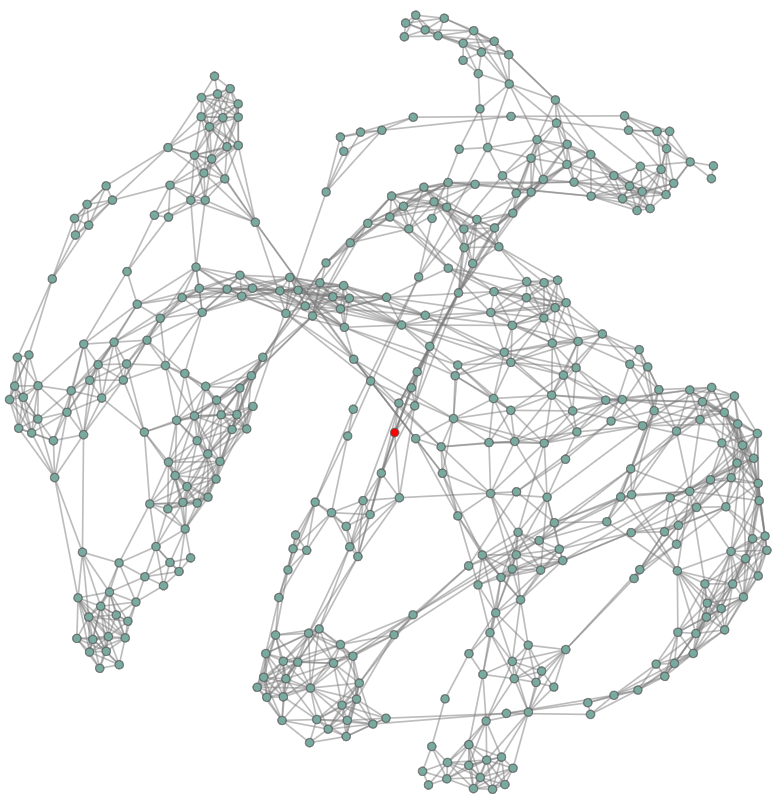}
	\caption{Network with node~$1$ in red.}
	\vspace{-20px}
	\label{fig:network}
\end{figure}

Another open issue in modeling epidemics over networks is the development of approximation schemes, whose aim is to describe the ``mean behavior'' of the disease evolution~\cite{pellis2015eight}.
One of their advantages is the ability to provide insights of the stochastic process without the need of interpolating large amounts of simulation results~\cite{rocha2016individual}.
We will show that the flexible nature of the Contagion Graph allows to formulate an approximate model based on statistical evidence.
\newline
Both applications of the Contagion Graph are shown by employing randomized simulations.

The remainder of this paper is structured as follows.
In Section~\ref{se:problem}, the dynamical system modeling the epidemic spread is presented.
We compare our model to the one of~\cite{koher2016infections} in Section~\ref{se:boolean}.
Section~\ref{se:contagion} presents the Contagion Graph, which allows for both complexity reduction during simulation and ``mean behavior'' analysis.
Concluding remarks are given in Section~\ref{se:conclusion}.


	\subsection{Notation}
Throughout this paper,
$\natnum$ denotes the set of nonnegative integers,
$\posint$ the set of positive integers,
and $\real$ the set of real numbers.
A graph is a pair $(\nodes,\arcs)$ where
$\nodes$ is the set of nodes and $\arcs$ is the set of arcs.
If the graph is undirected, $\arcs$ is the set of 
all two-element subsets $\{i,j\}$ of $\nodes$, so that there is an arc between node
$i$ and node $j$.
If the graph is directed,
$\arcs$ is the set of all pairs $(i,j)$, so that there is a
directed arc from node
$i$ to node $j$.

	\section{Problem Description} \label{se:problem}
\subsection{Multi-agent SIR model}
Consider $\nodes$, a set of
agents
labeled $1$ through $n\in\posint$
describing the whole population
in which an infection is spreading.
Each agent (also referred to as individual) 
can interact with other agents at discrete-time steps
(or iterations).
Assume the underlying network topology
modeling such interactions
to be an undirected graph, i.e., 
at every iteration $\tdisc$,
$\gr[k]:=(\nodes, \arcs[k])$. In particular,
the set of neighbors of agent $i\in\nodes$
at iteration $\tdisc$
is 
$$
	\neigh[k]:=\{j\in\nodes \mid \{j,i\}\in\arcs[k] \},
$$
and represents the set of all agents
that can interact with agent $i$
at iteration $k$.
Note that, by definition, self-arcs are not considered.
The infection can spread from one
individual (agent) to one or more of its neighbors in the graph $\gr[k]$,
and this way through the whole population (set).

To this end, let $\x[i][]:\natnum\mapsto\{0,1\}$ 
describe
if an agent, say $\aginset$,
is found to be \textit{infected}
at iteration $\tdisc$.
In particular,
$$
	\x=1
$$
means that agent $i$ is 
infected at iteration $k$,
otherwise $\x=0$.
In what follows, let
$\neighinf[k]\subseteq\neigh[k]$ be defined as
the subset of neighbors of
agent~$i$ at iteration~$k$
that are infected. Formally,
$$
	\neighinf[k]:=
	\{
		\aginset[j]
		\mid
		\{j,i\}\in\arcs[k],~
		\x[j]=1
	\}.
$$
If $\x=0$,
agent $i$ could be either \textit{susceptible}
(namely, it has never got in contact with the infection)
or \textit{recovered}
(namely, it has already got in contact with the infection, towards which 
it has developed immunity).
Let variable $\y$ capture whether
an individual~$i$ is recovered at iteration $k+1$,
i.e.,
$$
	\y=1
$$
means that agent $\aginset$
is recovered at iteration $k+1$,
otherwise $\y=0$, namely $i$ is \textit{susceptible} or \textit{infected}
at iteration $k+1$.
This clearly implies, $\forall\aginset$,
$\forall\tdisc$,
\eq\label{eq:yimpliesx}
	\y=1\implies\x[i][k+1]=0.
\qe
It is also assumed that immunity lasts forever,
i.e., $\forall\tdisc$,
\eq\label{eq:yincrease}
	\y[i][k+1]\geq\y.
\qe
Let, $\forall\aginset$, 
$\ki$ define the first time step
when agent~$i$ is found to be infected, i.e.,
\eq\label{eq:ki}
	\ki:=
	\inf
	\{
		\tdisc\mid
		\x=1
	\}.
\qe
In this paper, we use the convention that $\inf \emptyset = \infty$; therefore, if agent~$i$ does not get in contact
with the infection, then
$\ki=\infty$.
Agent $i$'s \textit{recovery time}
is the number of iterations from
$k_i+1$ until
agent~$i$ is recovered, 
and is denoted by $\rectime\in\posint$.
Formally, $\forall\aginset$, $\rectime$ is the smallest integer such that
\eq\label{eq:necRec}
	\y[i][\ki+\rectime]=1.
\qe
The system is heterogeneous, i.e.,
agents may have different
recovery times.

\subsection{Dynamics}

Consider a pair of agents,
i.e., $\{i,j\}\subset\nodes$,
such that $\{i,j\}\in\arcs[k]$ at a given iteration $k$.
Assume agent~$i$ is infected ($\x=1$)
whilst agent~$j$ is susceptible.
The probability that agent~$i$
infects its neighbor agent~$j$ at time $k$
is $p_{ij}(k)\in[0,1]$. Formally,
\eq
	\begin{split}	
	\mathrm{P}\Big(
		\x[j][k+1]=1
			\mid		
			\x[j]=0,~
			\x=1,~\\
			\y[j]=0,~
			\neighinf[k][j]=\{i\}
	\Big) = p_{ij}(k).
	\end{split}
\qe
\newline
%
We model this behavior by
defining
a random variable
drawn out of a Bernoulli distribution
with probability $p_{ij}(k)$, i.e.,
$$
	\forall\tdisc,~\forall\aginset,~\forall j\in\neigh,~
	\coeff\sim\mathcal{B}(p_{ji}(k)),
$$
called \textit{contagion coefficient}.
We assume that,
$\forall \{i,j\}\in\arcs[k]$,
$\coeff=\coeff[j][i]$.
If $\coeff[j][i]=1$,
an infected 
agent~$i$ infects a
susceptible neighbor agent~$j$
at iteration $k$.
Formally,
$\forall\tdisc$,
$\forall\{i,j\}\in\arcs[k]$,
\eq\label{eq:howinfection}
	\left. 
	\begin{aligned}
		\coeff[j][i] && =1\\
		\x && =1\\
		\y[j] && =0
	\end{aligned}
	\ \right\}
	~\implies
	\x[j][k+1]=1.
\qe
Consider, e.g., the case of an infection spreading across a population.
Each contagion coefficient models
whether two neighboring agents,
at a given time step,
have a contact, which would cause a contagion in case
one of the two is infected.
\newline
We also consider that
the infection can hit an individual
without being spread
from an infected agent in the network,
but rather by being transmitted from outside the network.
To this end,
let
$$
	\extinf=1
$$
denote that agent~$i$ gets in contact with the
infection (from outside of the network)
at time~$\tdisc$. 
Standard models considering closed 
populations (see, e.g.,~\cite{koher2016infections})
can be represented by having
$\extinf[i][0]=1$, with agent~$i\in\nodes$
defined traditionally as
\textit{patient zero}, and $u_j(k)=0$, $\forall j\in\nodes\setminus\{i\}$, $\forall k\in\natnum$.
Note that, $\forall\aginset$, $\forall\tdisc$,
\eq\label{eq:extinf}
\left. 
	\begin{aligned}
		\y &&=0\\
		\extinf &&=1
    \end{aligned}
    \ \right\}
	\implies
	\x[i][k+1]=1~.
\qe
\begin{definition}
	Agent~$i$ is \emph{protected against external infection}
	if, $\forall\tdisc$,
	$\extinf=0$.
	Otherwise, agent~$i$ is said to be
	\emph{subject to external infection}.
\end{definition}
Once an agent, say $\aginset$,
is infected, we know that 
it will recover in $\rectime$ iterations. 
We define a variable,
referred to as \textit{infection stopwatch},
formally defined as,
$\forall\aginset$, $\forall\tdisc$,
\eq\label{eq:si}
	s_i(k):=
	\begin{cases}
		0
		&\text{if }k\leq\ki\\
		\rectime + 1
		&\text{if }k\geq\ki+\rectime + 1\\
		k-\ki
		&\text{otherwise}
	\end{cases}~~.
\qe
Such a variable
counts the time steps since agent~$i$ got infected.
If no infection is developed up to time~$k$,
then this variable is~$0$.
If agent~$i$ has already recovered at time~$k$,
this variable equals $\rectime + 1$.
\newline
The dynamics of each
agent $i\in\nodes$ evolves
according to the
non-linear
discrete-time system
\begin{align}\label{eq:system}
	\begin{cases}
		x_i(k+1)
		=
		\varrho
		\left(R_i-s_i(k)\right)
		\varrho
		\Big(
		x_i(k) + \\ \qquad\qquad\qquad +
		\sum\limits_{j\in\neigh} 
		\coeff x_j(k)
		+\extinf
		\Big)
		\\
		s_i(k+1)
		=
		s_i(k) 
		+ 
		x_i(k)\\
		y_i(k) \phantom{{}+{}1}= 1 - \varrho(\rectime - s_i(k))
	\end{cases}
\end{align}	
where $\varrho:\real\mapsto\{0,1\}$
is the step function
\[
\varrho(\circ) =
\begin{cases}
0 & \mbox{ if }\circ\leq 0\\
1 & \mbox{ if }\circ>0
\end{cases}.
\]
In what follows,
let
$\xvec$ and $\svec$ be $n$-dimensional
vectors stacking, respectively,
the infection variables and
the infection stopwatch variables
of all agents at time step $k$, i.e.,
$\forall\tdisc$,
$\forall i\in\{1,\dots,n\}$,
$
	[\xvec]_i=\x,~
	[\svec]_i=\s~.
$
\begin{proposition}
		System (\ref{eq:system}) 
		guarantees that properties
		(\ref{eq:yimpliesx}), (\ref{eq:yincrease}), (\ref{eq:necRec}), (\ref{eq:howinfection})-(\ref{eq:si})
		are satisfied, $\forall\aginset$, $\forall\tdisc$.		
	\begin{proof}
		{The proof is omitted due to space limitation}.
	\end{proof}
\end{proposition}

\begin{remark}
	Note that the state-space cardinality for system~\eqref{eq:system} is $2^n \prod_{i\in\nodes}(\rectime + 2)$.
\end{remark}

\begin{example}\label{ex:time-based}
	Consider the static network $(\nodes, \arcs)$ 
	depicted in \emph{Fig.~\ref{fig:network}},
	in which an infection is spreading with dynamics~(\ref{eq:system}).
	Agent~$1$ gets in contact with the infection
	at time~$0$, i.e.,
	$\extinf[1][0]=1$. We simulate three different scenarios
	in order to address the impact that different parameters have
	on the spreading of the infection.
	\newline 
	Scenario 1:
	let, $\forall(j,i)\in\arcs$, $\forall k\in\natnum$, $p_{ij}(k)=0.2$ and 
	$\forall\aginset$, $\rectime\in[3,5]$ (randomly extracted from this set). We can observe the infection
	dynamics in \emph{Fig.~\ref{fig:epidemics_pij_02_Rimax_5}}. Some agents never get infected. In fact,
	after time $k=40$, a subset remains susceptible and
	the infection disappears.	
	\newline 
	Scenario 2:	in this case, we increment the maximum recovery time to $30$ iterations.
	This implies, as in \emph{Fig.~\ref{fig:epidemics_pij_02_Rimax_30}}, that 
	the peak of infection is wider and longer-lasting. 
	Also, the whole population gets in contact with the disease.
	\newline 
	Scenario 3: 
	on the other hand, if we increment the
	possibility of infecting another agent (namely, $\forall(j,i)\in\arcs$, $\forall k\in\natnum$, $p_{ij}(k)=0.5$),
	but we keep the maximum recovery time as $5$ iterations,
	we obtain a fast spread of the infection, in which the whole population
	gets in contact with the disease, but the infection disappears after $30$ steps (\emph{Fig.~\ref{fig:epidemics_pij_05_Rimax_5}}).
\begin{figure}
	\includegraphics[width=\columnwidth]{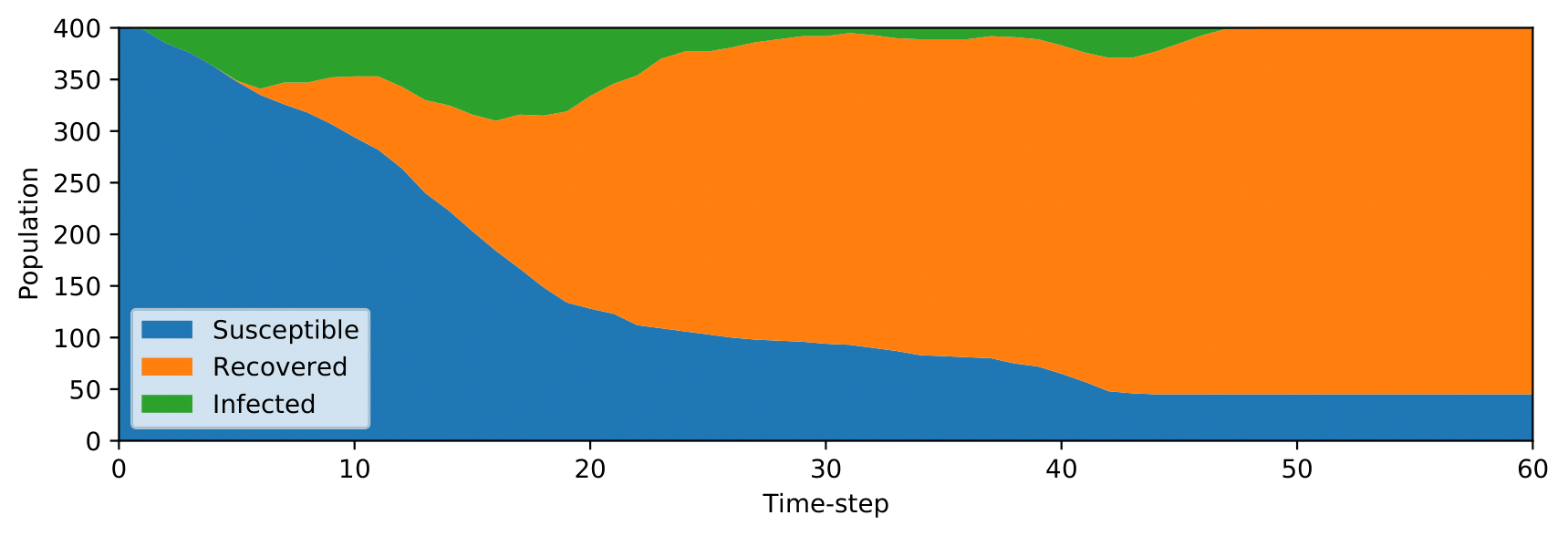}
	\caption{Infection spreading with $p_{ij}=0.2$ and $\rectime\in[3,5]$.}
	\label{fig:epidemics_pij_02_Rimax_5}
\end{figure}
\begin{figure}
	\includegraphics[width=\columnwidth]{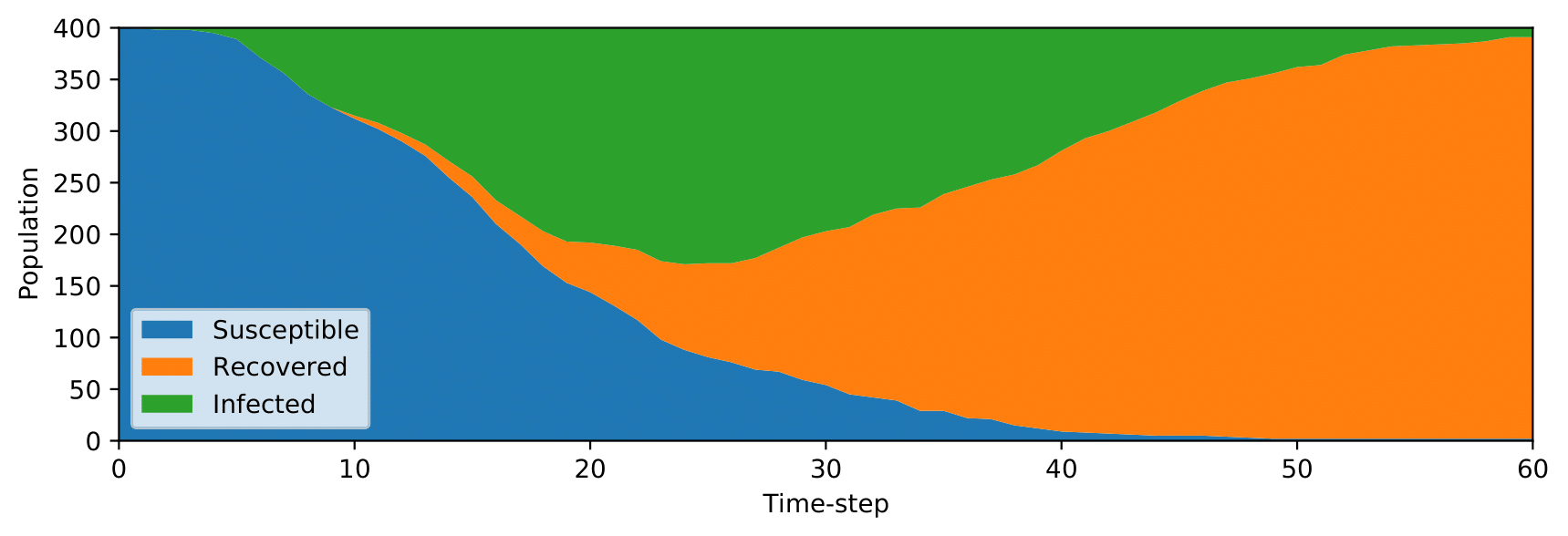}
	\caption{Infection spreading with $p_{ij}=0.2$ and $\rectime\in[3,30]$.}
	\label{fig:epidemics_pij_02_Rimax_30}
	\end{figure}
\begin{figure}
	\includegraphics[width=\columnwidth]{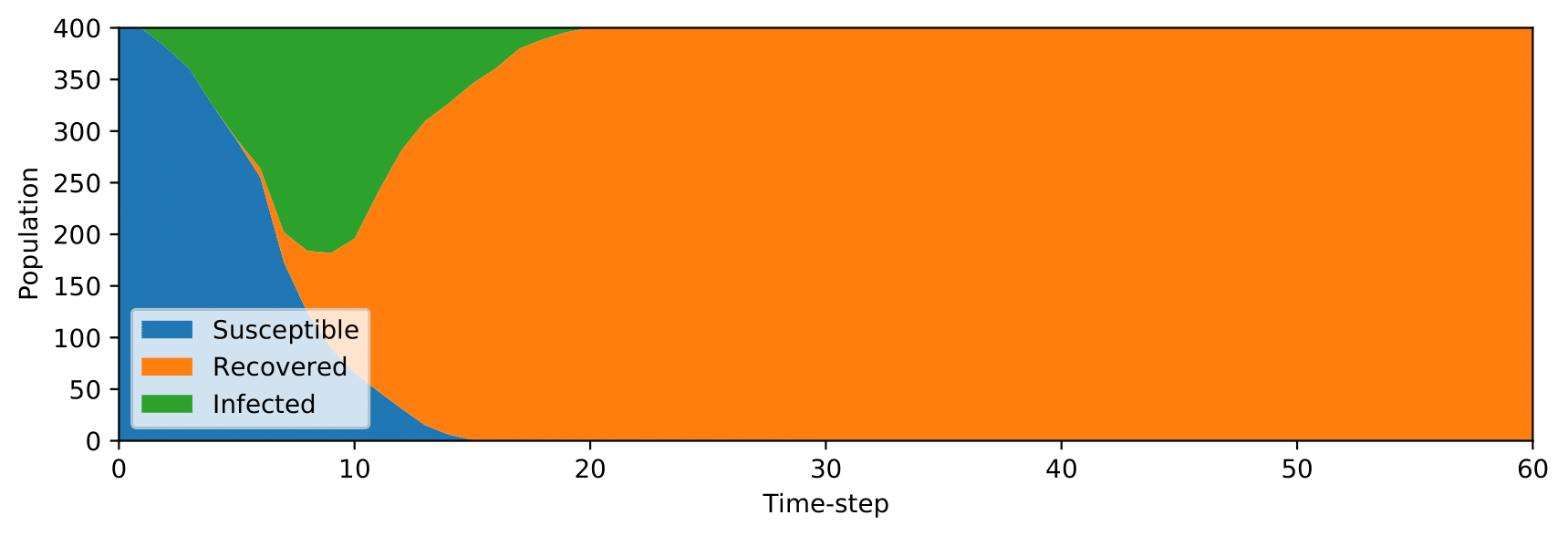}
	\caption{Infection spreading with $p_{ij}=0.5$ and $\rectime\in[3,5]$.}
	\label{fig:epidemics_pij_05_Rimax_5}
\end{figure}
\end{example}
Running each one of these simulation
takes on average $0.95\mathrm{s}$ on a PC with an Intel i7 processor at 2.90GHz with Python 3.6.
The running time increases with the simulation horizon, which must be set large enough to capture the entire disease evolution.
As shown in the examples, an adequate choice of the simulation horizon greatly depends on parameters of individuals such as $p_{ij}(k)$ and $\rectime$.
This proves that, despite the benefits of agent-based models,
the time needed for simulation not only increases with network size, but also depends on agent parameters.
As motivated by~\cite{pellis2015eight}, this paper
is concerned with finding a formal method to run 
multi-agent simulation on networks 
with low computational effort.
\begin{remark}
	Note that $\x$, $\y$, $\ki$, and $\s$
	are \emph{random variables}, as they depend on contagion coefficients $\coeff$.
\end{remark}

	\section{Model Dynamics with Boolean Algebra} \label{se:boolean}
In the following, we show that system~\eqref{eq:system}
can be rewritten
as a dynamical system in the Boolean algebra.
This proves that our model is 
a stochastic dynamical
generalization of~\cite[Eq. (3)]{koher2016infections}.
In fact, our system
is represented in the formalism
of dynamical systems, whilst 
\cite[Eq. (3)]{koher2016infections}
has a state vector whose dimension increases
with iterations. Moreover, in~\cite[Eq. (3)]{koher2016infections}
contagion is deterministic, whereas in~\eqref{eq:system}
it depends, at every time step $\tdisc$,
on the stochastic realizations of
variables $\{\coeff\}_{\{i,j\}\in\arcs[k]}$.
Moreover,
\cite[Eq. (3)]{koher2016infections} does 
not consider infections originating outside the considered network, but
only epidemics within a closed population.
\begin{remark}
	Due to space limitation,
	in this section
	we consider~\eqref{eq:system}
	without the presence of $\extinf$.
\end{remark}
\subsection{Fundamentals of Boolean Algebra}
Consider two Boolean variables $a$ and $b$.
We define, respectively, disjunction, conjunction
and negation operations as,
$\forall a\in\{0,1\}$,
$\forall b\in\{0,1\}$,
\[
	a\vee b = \max(a,b),\quad
	a\wedge b = \min(a,b),\quad
	\neg a = 1 - a.
\]
We extend the operations to matrices.
To this end, consider three Boolean matrices, i.e.,
$A,B\in\{0,1\}^{n\times m}$, $C\in\{0,1\}^{m\times p}$.
Disjunction, conjunction,
and negation operations are defined as
\begin{align*}
	[A \vee B]_{ij} &= a_{ij} \vee b_{ij} &i=1\dots n,~j=1\dots m,\\
	[A \wedge C]_{ij} &= \bigvee_{h=1}^m a_{ih} \wedge b_{hj}&i=1\dots n,~j=1\dots p,\\
	[\neg A]_{ij} &= \neg a_{ij}&i=1\dots n,~j=1\dots m.
\end{align*}
Furthermore, let us define the element-wise conjunction 
operation (Hadamard product in the Boolean algebra)
$\circleland$, i.e., 
\[
	[A \circleland B]_{ij} = a_{ij} \wedge b_{ij}, 
	\qquad i=1\dots n,~j=1\dots m.
\]
\subsection{Infection Model in The Boolean Algebra}
We define the following dynamical system in the Boolean
algebra with $\tdisc$ the iteration index:
\begin{align}\label{eq:system_bool} 
	\begin{cases}
		\xbo[k+1]
		&=
		\left(\coeffmatrix(k) \wedge \xbo\right) \circleland (\neg \ybo)
		\\
		\ybo[k+1]
		&=
		\ybo
		\vee 
		\xbo[{k-\rectime[]+1}]
	\end{cases},
\end{align}
where
\begin{itemize}
	\item $\xbo\in\{0,1\}^n$ is the \textit{Boolean infection vector},
	such that $\xbo[k][i]$ is $1$ if individual $i$ is infected at
	time~$k$, $0$ otherwise;
	\item $\ybo\in\{0,1\}^n$ is the \textit{Boolean recovery vector},
	such that $\ybo[k][i]$ is $1$ if individual $i$ is recovered at
	time~$k$, $0$ otherwise;
	\item $\coeffmatrix(k)\in\{0,1\}^{n\times n}$, such that
	\begin{align*}
		\coeffmatrix_{ij}(k):=
		\begin{cases}
			\coeff&\text{if }\{i,j\}\in\arcs[k]\\
			1&\text{if }i=j\\
			0&\text{otherwise}
		\end{cases}\qquad;
	\end{align*}
	\item $\xbo[{k-\rectime[]+1}]$ is short hand notation for the 
	$n$-dimensional vector
	of elements $\left[\xbo[{k-\rectime[]+1}]\right]_i= \xbo[{k-\rectime+1}][i]$,
	$i=1\dots n$.
\end{itemize}
To initialize~\eqref{eq:system_bool} we define $\xbo[k][i]=0$, $\forall k<0$.
Note that, if 
$\ybo[0]=0$, the
second equation of system~\eqref{eq:system_bool}
can be written as
\[
	\ybo[k+1] = \bigvee_{t=0}^{k-\rectime+1} \xbo[t].
\]
In this case, 
and considering all contagion coefficients equal to one,
i.e., $\forall\{i,j\}\in\arcs[k]$, $\coeff=1$,
it is immediate to see the equivalence
of system~\eqref{eq:system_bool}
with \cite[Eq.~(3)]{koher2016infections}.
\begin{remark}
    The state of system~\eqref{eq:system_bool} at iteration $k$ is
    \begin{align*}
    	\mathbf{x}_s^b(k) := [& \xbo[k][1],\dots,\xbo[{k-\rectime[1]+1}][1],\dots,\\
    	& \xbo[k][n],\dots,\xbo[{k-\rectime[n]+1}][n],\\
    	&\ybo[k][1],\dots,\ybo[k][n]]^\top \in\{0,1\}^{n_s},
    \end{align*}
    with $n_s = n + \sum_{i\in\nodes} \rectime$.
	Unlike~\cite[Eq.(3)]{koher2016infections},
	in system~\eqref{eq:system_bool} the state dimension
	does not grow with time, but has dimension
	$n_s$. 
\end{remark}

\begin{remark}
	The state-space cardinality for system~\eqref{eq:system_bool} is $2^{n_s}$.
	It is immediate to show that system~\eqref{eq:system} has a lower state-space cardinality than system~\eqref{eq:system_bool} if 
	\[ 
        \sum_{i\in\nodes} \log_2(\rectime + 2) \leq \sum_{i\in\nodes} \rectime.
	\]
\end{remark}

In what follows,
we prove that~\eqref{eq:system_bool}
can be seen as
the Boolean formulation
of~\eqref{eq:system}.
\begin{proposition}
Assume, $\x[i][0]=\xbo[0][i]$, $\s[i][0]=\ybo[0][i]=0$, 
and, by definition, $\xbo[k][i]=0\quad  \forall k<0$.
%
Vectors $\xvec$ and $\svec$,
respectively $\xbo$ and $\ybo$,
evolving according to~\eqref{eq:system},
respectively~\eqref{eq:system_bool},
satisfy, $\forall\tdisc$, $\forall\aginset$,
\eq\label{eq:indStat}
	\x = \xbo[k][i] \quad 
	\text{and}\quad
	\s\geq \rectime \iff
	\ybo[k][i]=1.
\qe
\end{proposition}
\begin{proof}
	The proof follows by strong induction.
	\newline
	Let the \textit{induction statement} 
	be~(\ref{eq:indStat}).
	The base case, for $k=0$,
	is trivially verified
	in the hypothesis.
	\newline
	The \textit{inductive hypothesis} is
	that, $\forall k\leq h,~h\in\natnum$,
	(\ref{eq:indStat}) is true.
	Thus, 
	we need to prove
	that 
	(\ref{eq:indStat}) must also be true for $k=h+1$.
%
	\newline
	First of all, let us note that, by the inductive hypothesis, $\forall\aginset$,
	
	{\footnotesize
	\begin{align}
	    \varrho(
	    \x[i][h] + \sum_{j\in\neigh[h]} \coeff[i][j][h] \x[j][h]
	    )
	    &= \left(
	    \x[i][h] \vee
	    \bigvee_{j\in\neigh[h]} \coeff[i][j][h] \wedge \x[j][h]
	    \right)
	    \nonumber \\
	    &= \left[
	    \coeffmatrix(h) \wedge \xvec[h]
	    \right]_i. \label{eq:aux1}
	\end{align}
	}%
	By the induction hypothesis
	and by definition of $\varrho(\circ)$,
	we have, $\forall\aginset$,
	$
		\varrho(\rectime-s_i(h))=0\iff\ybo[h][i]=1,
	$
	equivalently, $\forall\aginset$,
	\eq\label{eq:aux2}
		\varrho(\rectime-s_i(h))=\neg\ybo[h][i],
	\qe
	By bringing together (\ref{eq:aux1}) and (\ref{eq:aux2}),
	we conclude that,
	by the inductive hypothesis, $\forall\aginset$,
	\eq\label{eq:firstIndProven}
		\x[i][h+1]=\xbo[h+1][i].
	\qe
	
	At this point, to conclude the proof,
	we need to prove that the induction hypothesis
	also implies, $\forall\aginset$,
	\eq\label{Eq:secondIndToPRove}
		\s[i][h+1]\geq \rectime \iff
		\ybo[h+1][i]=1.
	\qe
	To this end,
	note that,
	by~(\ref{eq:system}), $\forall\aginset$,
	\begin{multline}\label{eq:secInd1}
		\s[i][h+1]\geq\rectime
		\iff
		\\ \x[i][\ell-\rectime+1]=1\quad \mbox{for some }\ell\in\{\rectime-1,\dots,h\}
		~.
	\end{multline}
	By the induction hypothesis, $\forall\aginset$,
	\eq\label{eq:secInd2}
		\x[i][\ell-\rectime+1]=\xbo[\ell-\rectime+1][i].
	\qe
	By (\ref{eq:system_bool}), $\forall\aginset$,
	\eq\label{eq:secInd3}
		\xbo[\ell-\rectime+1][i]=1
		\implies
		\ybo[h+1]=1~,
	\qe
	for some
	$\ell\in\{\rectime-1,\dots,h\}$.
	By bringing together (\ref{eq:secInd1}), 
	(\ref{eq:secInd2}), and (\ref{eq:secInd3}),
	under the induction hypothesis,
	(\ref{Eq:secondIndToPRove}) is verified.
	This proves (\ref{eq:indStat}) for $k=h+1$,
	thus implying, by strong induction, 
	that  (\ref{eq:indStat}) always holds.
	This concludes the proof.
\end{proof}

    \section{Contagion Graph} \label{se:contagion}
As discussed in literature and
seen in Example~\ref{ex:time-based},
an important issue for models
of epidemics over networks is
the computational burden.
In this section,
we investigate a 
possible way to obtain
simulation results for
system~(\ref{eq:system}), and therefore~\eqref{eq:system_bool}, with 
reduced computational complexity.
\subsection{Contagion Graph as equivalent model} \label{su:contagion_equivalent}
To simplify exposition,
consider the following two assumptions that
will be relaxed in future work.
\begin{assumption}
	\label{en:assumption1}
	The network topology is
	constant over time, i.e.,
	$\forall k\in\natnum$,
	$\arcs[k]=\arcs[]$.
\end{assumption}
\begin{assumption}
	\label{en:assumption2}
	The stochastic process $\{\coeff\}_{k\in\natnum}$ is stationary,
	i.e.,
	$\forall k\in\natnum$, $\forall\{i,j\}\in\arcs[]$, $p_{ij}(k)=p_{ij}$.
%
%
\end{assumption}
%
Assumption~\ref{en:assumption1}
corresponds to the case in which
neither restricting measures are taken
in order to contain the spread of the infection,
nor new connections between
individuals are established.
With Assumption~\ref{en:assumption2},
we consider the probability of
interaction between any pair of
individuals to be constant over time.

Consider $\ki$ as defined in (\ref{eq:ki}).
Set $\{\ki \}_{\aginset}$
%
provides a full
description of the evolution of the epidemic; 
if $\inftime\in\natnum$, individual
$i$ gets infected at time $\inftime$ 
and recovers
at time $\inftime+\rectime+1$,
whilst, if 
$\inftime=\infty$, agent~$i$ never gets infected.
\begin{definition}\label{def:tforinf}
	The random variable
	$\tforinf(k)$ denotes, $\forall\{i,j\}\in\arcs$,
	the number of time-steps that it takes agent~$j\in\neigh[]$
	to infect its neighbor~$i$,
	assuming $\inftime[j]=k$, and ignoring the effect of other
	individuals.
	Formally, 
	\eq \label{eq:tforinf}
		\tforinf(k) := 1+\inf\{ h\in\{k,\ldots,k+\rectime[j]-1\} \ |\ 
		\coeff[i][j][h]=1 \}.
	\qe
\end{definition}
\begin{definition}\label{def:tauOpenInf}
	Given an agent~$i$,
	variable $\tforex\in\natnum$ denotes the earliest time of a
	possible external infection, i.e.,
	$$
		\tforex := 1+\inf\{\tdisc\mid \extinf=1\}.
	$$
\end{definition}
It is clear by definition that
agents protected against external infections have
$\tforex=\infty$.
\begin{definition}\label{def:setExtInf}
	Let $\exnodesnotlab\subseteq\nodes$
	denote the set of agents subject to
	external infections, i.e.,
	$$
		\exnodesnotlab :=\{
			\aginset\mid
			\tforex\in\natnum
		\}.
	$$
	Let now $\exnodes$ be a new set of
	nodes labeled with negative numbers,
	such that each node in $\nodes^{\mathrm{ext}}$
	corresponds to one node in $\exnodes$, i.e.,
	$$
		\exnodes:=
		\{
			-i
			\mid
			i\in\exnodesnotlab
		\}.
	$$
\end{definition}
\begin{lemma}\label{lemma:stat}
	Under Assumptions 1-2,
	$\tforinf(k)$ is a stationary process
	with probability distribution
	\eq
		\label{eq:SIR_distrib}\small 
		\probab{\tforinf= \tau}=
		\begin{cases}
			p_{ij}(1-p_{ij})^{\tau-1} &
			\mbox{if } \tau \in\{1,\ldots,\rectime[j]\}\\
			(1-p_{ij})^{\rectime[j]} &
			\mbox{if }\tau=\infty\\
			0 &
			\mbox{otherwise}
		\end{cases}.
	\qe	
	\begin{proof}
		If $\tau\in\{1,\ldots,\rectime[j]\}$, then
		$\tforinf(k)=\tau$ corresponds to event
		$\coeff[i][j][k]=\ldots=\coeff[i][j][k+\tau-1]=0$,
		$\coeff[i][j][k+\tau]=1$.
		Since $\coeff$ is a Bernoulli distributed random
		variable, in this case $\tforinf(k)$ assumes
		the geometric distribution $\mbox{Geo}(p_{ij})$,
		as in the first line of~\eqref{eq:SIR_distrib}.
		Note that for
		$\tau\in\posint\setminus\{1,\dots,\rectime[j]\}$,
		$\tforinf(k)=\tau$ has zero probability, since it
		would imply $h>k+\rectime[j]-1$ in~\eqref{eq:tforinf}.
		Therefore, event $\coeff=\ldots=\coeff[i][j][k+R_j-1]=0$,
		which has probability $(1-p_{ij})^{\rectime[j]}$, corresponds
		to
		\begin{align*}
			\tforinf(k) &= 1 + \inf_{\substack{h\in \{0,\ldots,\rectime[j]-1\}\cap (\posint \setminus\{1,\dots,\rectime[j]\})  \\
					\coeff[i][j][k+h]=1}} h\\
			&=
			1 + \inf_{h\in\emptyset} h = \infty. 
		\end{align*}
		This concludes the proof.
	\end{proof}
\end{lemma}
\begin{theorem}\label{thm:eqdistki}
	Under Assumptions~\ref{en:assumption1}-\ref{en:assumption2},
	for any agent~$i$ 
	protected against external infections,
	we have\footnote{
		The symbol $\eqdist$
		denotes equality in distribution.
	}
	\eq\label{eq:inftime}
		\inftime \eqdist \inf_{j\in\neigh[]} \left(
		\inftime[j] + 
		\phi_{\rectime[j]}\left(
			\ceil{\log_{1-p_{ij}} u_{ij}}
		\right)
		\right),
%
	\qe
	where
	\[
		\phi_{\rectime[j]}(\circ) :=
		\begin{cases}
			\circ & \mbox{if } \circ \leq \rectime[j] \\
			\infty & \mbox{otherwise}
		\end{cases},
	\]
	and
	$u_{ij}\sim\pazocal{U}(0,1)$.
	\begin{proof}
		By incorporating~(\ref{eq:howinfection})~into~(\ref{eq:ki}),
		for~$i$ being an agent protected against external infection,
		one has 
		\eqs
			\ki=\inf_{j\in\neigh[]}
			\{
				\tdisc\mid
				\coeff[j][i][k-1]=1,~
				\x[j][k-1]=1
			\}~.
		\qes
		Also by~(\ref{eq:ki}) and definition
		of recovery time,
		the latter is equivalent to
		\eqs
			\ki=\inf_{j\in\neigh[]}
			\{
				\tdisc\mid
				\coeff[j][i][k-1]=1,~
				1\leq k-\ki[j]\leq\rectime[j]+1
			\}~,
		\qes
		that, by considering (\ref{eq:tforinf}), 
		can be rewritten as
		\eqs
			\ki\eqdist
			\inf_{j\in\neigh[]}
			\{
				\ki[j]
				+
				\tforinf(\ki[j])
			\}~.
		\qes
		By~Lemma~\ref{lemma:stat},
		being $\tforinf$ stationary,
		\eq\label{eq:kialmostdone}
			\ki\eqdist
				\inf_{j\in\neigh[]}
			\{
			\ki[j]
			+
			\tforinf
			\}~.
		\qe
		A realization of $\tforinf$, by~\eqref{eq:SIR_distrib},
		is determined by applying a threshold $\phi_{\rectime[j]}(\cdot)$
		on a realization of a random variable, say $\tilde{\tau}_{ij}$, with
		geometric distribution, i.e., $\ttilde\sim\mbox{\normalfont Geo}(p_{ij})$,
		An algorithm for sampling a geometric random variable
		in constant time\footnote{
			Hereafter, we assume that we
			can draw uniform distributed random variables
			in constant time.} is given in literature,
		see, e.g.,~\cite[section~X.2]{devroye1986non}.
		Given a realization of a uniformly distributed random
		variable $u_{ij}\sim\pazocal{U}(0,1)$, a sample of $\ttilde$ is
		is computed exactly using
		formula $\ceil{\log_{1-p_{ij}} u_{ij}}$,
		where $\ceil{\cdot}$ is the ceiling
		function.
		By incorporating 
		\eq\label{eq:howcalculatetauij}
			\tforinf \eqdist \phi_{\rectime[j]}\left(
			\ceil{\log_{1-p_{ij}} u_{ij}}
			\right)
		\qe
		into (\ref{eq:kialmostdone}),
		the proof is concluded.
		%
	\end{proof}
\end{theorem}
\begin{corollary}
	Under Assumptions~\ref{en:assumption1}-\ref{en:assumption2},
	for any agent~$i$ 
	protected against external infections,
	we have	
	\eq\label{eq:exinftime}
		\inftime \eqdist \inf \left(
		\inf_{j\in\neigh[]} \left(
		\inftime[j] + 
		\phi_{\rectime[j]}\left(
		\ceil{\log_{1-p_{ij}} u_{ij}}
		\right)
		\right),~
		\tforex \right),
	\qe
	for 
	$u_{ij}\sim\pazocal{U}(0,1)$.
	\begin{proof}
		Consider~(\ref{eq:kialmostdone}).
		In case of an agent subject to external infection,
		we could rewrite it as
		\eq\label{eq:kialmostdone_open}
			\ki\eqdist
			\inf
			\left\{
			\inf_{j\in\neigh[]}
			\{
			\ki[j]
			+
			\tforinf
			\},~\tforex
			\right\}~.
		\qe
		By incorporating~(\ref{eq:howcalculatetauij})
		into the latter, the proof immediately follows.
	\end{proof}
\end{corollary}
We have gathered all notions needed for
defining the \textit{Contagion Graph}
associated to network $(\nodes,\arcs)$
in which an infection is spreading across
nodes with dynamics modeled by~(\ref{eq:system}).
\begin{definition}
	\label{def:CG}
	The \textbf{Contagion Graph} of network $(\nodes, \arcs)$
	with infection dynamics~(\ref{eq:system})
	is a directed graph with random weights
	whose nodes are $\nodes\cup\exnodes$.
	The set of arcs is $\arcs_1\cup\arcs_2$, where
	$$
		\arcs_1:=\{(j,i)\mid\{i,j\}\in\arcs,~ \tforinf\in\natnum \}
	$$
	and
	$$
		\arcs_2:=\{
			(-i,i)
			\mid
			i\in\exnodesnotlab
		\}.
	$$
	Arc weights are,
	$\forall (j,i)\in\arcs_1$, $\tforinf$, 
	and,
	$\forall (-i,i)\in\arcs_2$,
	$\tforex$.
	Weights $\tforinf$
	are random variables
	with distributions as in~(\ref{eq:howcalculatetauij}).
\end{definition}
\begin{example} \label{ex:1} 
\begin{figure}
\begin{subfigure}{.35\linewidth}
    \centering
    \resizebox{1\linewidth}{!}{
\begin{tikzpicture}[
            > = stealth, 
            shorten > = 1pt, 
            auto,
            node distance = 2cm, 
            semithick 
        ]
\tikzstyle{every state}=[
            draw = black,
            thick,
            fill = white,
            minimum size = 7mm
        ]
\newcommand{\sizedot}{2};

\node[state] (1) {$1$};
\node[state] (2) [right of=1] {$2$};
\node[state] (3) [right of=2] {$3$};
\node[state] (4) [below of=3] {$4$};
\node[state] (5) [left of=4] {$5$};

\path (2) edge (1);
\path (2) edge (3);
\path (3) edge (4);
\path (5) edge (4);
\path (5) edge (2);

\end{tikzpicture}
    }
    \caption{Graph topology.}
    \label{fig:ex_topology}
\end{subfigure}
\begin{subfigure}{.6\linewidth}
    \centering
    \resizebox{1\linewidth}{!}{
\begin{tikzpicture}[
            > = stealth, 
            shorten > = 1pt, 
            auto,
            node distance = 2cm, 
            semithick 
        ]
\tikzstyle{every state}=[
            draw = black,
            thick,
            fill = white,
            minimum size = 7mm
        ]
\newcommand{\sizedot}{2};

\node[state] (1) {$1$};
\node[state] (2) [right of=1] {$2$}; 
\node[state] (3) [right of=2] {$3$};
\node[state] (4) [below of=3] {$4$};
\node[state] (5) [left of=4] {$5$}; 

\node[state,fill=red!60] (11) [left of=1] {$\exi[1]$};
\node[state,fill=red!60] (44) [right of=4] {$\exi[4]$};

\path[->] (2) edge[bend left] node {$\tforinf[1][2]$} (1);
\path[->] (1) edge[bend left] node {$\tforinf[2][1]$} (2);
\path[->] (2) edge[bend left] node {$\tforinf[3][2]$} (3);
\path[->] (3) edge[bend left] node {$\tforinf[2][3]$} (2);
\path[->] (3) edge[bend left] node {$\tforinf[4][3]$} (4);
\path[->] (4) edge[bend left] node {$\tforinf[3][4]$} (3);
\path[->] (4) edge[bend left] node {$\tforinf[5][4]$} (5);
\path[->] (5) edge[bend left] node {$\tforinf[4][5]$} (4);
\path[->] (5) edge[bend left] node {$\tforinf[2][5]$} (2);
\path[->] (2) edge[bend left] node {$\tforinf[5][2]$} (5);

\path[->] (11) edge node {1} (1);
\path[->] (44) edge node {3} (4);

\end{tikzpicture}
    }
    \caption{Contagion Graph.}
    \label{fig:ex_contagion}
\end{subfigure}

\vskip10pt

\begin{subfigure}{1\linewidth}
    \centering
    \resizebox{!}{.1\textheight}{
\begin{tikzpicture}[
            > = stealth, 
            shorten > = 1pt, 
            auto,
            node distance = 2cm, 
            semithick 
        ]
\tikzstyle{every state}=[
            draw = black,
            thick,
            fill = white,
            minimum size = 7mm
        ]
\newcommand{\sizedot}{2};

\node[state] (1) {$1$};
\node[state] (2) [right of=1] {$2$};
\node[state] (3) [right of=2] {$3$};
\node[state] (4) [below of=3] {$4$};
\node[state] (5) [left of=4] {$5$};

\node[state,fill=red!60] (11) [left of=1] {$\exi[1]$};
\node[state,fill=red!60] (44) [right of=4] {$\exi[4]$};

\path[->] (1) edge[bend left] node {3} (2);
\path[->] (2) edge[bend left] node {2} (1);
\path[->] (2) edge node {3} (3);
\path[->] (3) edge node {2} (4);
\path[->] (4) edge[bend left] node {1} (5);
\path[->] (5) edge[bend left] node {2} (4);
\path[->] (5) edge node {3} (2);

\path[->] (11) edge node {1} (1);
\path[->] (44) edge node {3} (4);

\end{tikzpicture}
    }
    \caption{A possible Contagion Graph realization.}
    \label{fig:ex_contagion_real}
\end{subfigure}
\caption{Graphs of Example \ref{ex:1}.}
\end{figure}
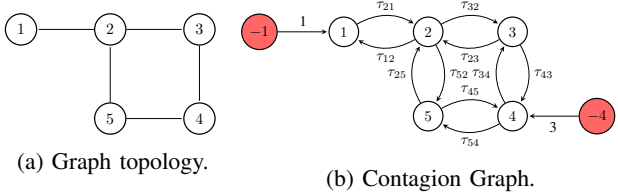
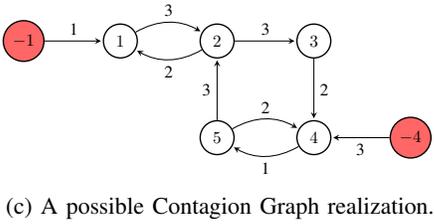

Consider a network of $5$ agents
with infection dynamics~(\ref{eq:system})
as in \emph{Fig.~\ref{fig:ex_topology}}.
Assume that agents $1$ and $4$ are subject to external infection.
The recovery time is supposed to 
be $3$ time steps for each individual, i.e.,
$\forall i\in\nodes$, $\rectime = 3$, and the time
of infection from external sources are $\tforex[1]=1$,
$\tforex[4]=3$.
From these parameters, it is possible to draw the 
Contagion Graph
in \emph{Fig.~\ref{fig:ex_contagion}}, where nodes of set
$\nodes=\{1,2,3,4,5\}$ are in white and the ones of set
$\exnodes=\{ \exi[1],\exi[4]\}$ are in red; note that weights
$\tforinf$ are random variables, while $\tforex$ are
deterministic.
Assuming that, $\forall \{i,j\}\in\arcs$,
$p_{ij}=0.2$,
a realization for random variables $\tforinf$
is computed as in (\ref{eq:howcalculatetauij}).
The corresponding realization of the Contagion Graph is
shown in \emph{Fig.~\ref{fig:ex_contagion_real}}.
Here, the arcs that are missing from the Contagion Graph
correspond to a realization $\tforinf=\infty$.
\end{example}
The Contagion Graph provides an efficient
method to compute all $\{\ki \}_{\aginset}$,
as explained by the following result.
%
%
\begin{theorem}\label{thm:minPath}
	For each agent $\aginset$,
	the random variable $\ki$
	is equal in distribution to the 
	\textit{minimum weighted path} on
	the Contagion Graph from any node in $\exnodes$
	to~$i$. 
	\begin{proof}
		Consider~(\ref{eq:kialmostdone_open}).
		This can be expanded as, $\forall\aginset$,
		\begin{align*}
			\ki\eqdist&
			\inf
			\left\{
			\inf_{\substack{j\in\neigh[]\\\ell\in\neigh[][j]}}
			\left\{
			\ki[\ell]
			+
			\tforinf[j][\ell]+\tforinf
			,~\tforex[j]+\tforinf\right\},~\tforex
			\right\}\\
			\eqdist&
			\inf_{\substack{j\in\neigh[]\\\ell\in\neigh[][j]}}
			\left\{
			\ki[\ell]
			+
			\tforinf[j][\ell]+\tforinf
			,~\tforex[j]+\tforinf,
			~\tforex
			\right\}~.
		\end{align*}
		One can see that $\ki$ is equal in distribution
		to the minimum 
		between
		the path from $-i\in\exnodes$ to $\aginset$, 
		the path from $-j\in\exnodes$ to $\aginset$,
		and  $k_\ell$ plus the path from $\ell\in\nodes$ to $\aginset$.		
		By doing this recursively,
		one obtains
		\eq\label{eq:pathminweight}
			\inftime \eqdist \inf_{\substack{\{h_1,\dots,h_\ell\}\subseteq\nodes\\h_\ell=i}}
			\left\{ \tforex[{h_1}] +
			\sum_{m=2}^{\ell}
			\tforinf[h_{m}][h_{m-1}]\right\}~,
		\qe
		which is, by Definition~\ref{def:CG},
		the minimum path going from one
		node $-h_1\in\exnodes$ to $\aginset$.
	\end{proof}
\end{theorem}
By this latter Theorem,
in order to simulate the epidemics'~dynamics~(\ref{eq:system})
over a network~$(\nodes,\arcs)$, 
it is sufficient
to obtain a realization of the Contagion Graph and
compute~(\ref{eq:pathminweight}) for each node.
As it will be
shown in Section~\ref{sec:compCompl}, this allows to decrease 
the computational complexity if compared to 
running~(\ref{eq:system}).
\newline

\begin{corollary}\label{coro:wave}
	For a realization of the Contagion Graph,
	we can compute the number of infected agents
	at every time~$\tdisc$ as
	\eqs
		\sum_{i\in\nodes} x_i(k)
		=
		\sum_{i\in\nodes}
		\varrho(k-\inftime+1)
		-
		\varrho(k-\inftime-\rectime).
	\qes
	%
	\begin{proof}
		{The proof is omitted due to space limitation}.
	\end{proof}
\end{corollary}
\begin{example}[Continuation of Example~\ref{ex:1}]
	\label{ex:2}
	Given the realization of \emph{Fig.~\ref{fig:ex_contagion_real}},
	we can determine the evolution of the disease spread by
	computing the paths of minimum weight from any node of
	$\exnodes$ to any node of $\nodes$.
	By solving the single-source shortest path problem from
	nodes $\exi[1]$ and $\exi[4]$, and then taking the path
	of minimum weight between the two, we get $
	\inftime[1]=1,\ 
	\inftime[2]=4,\ 
	\inftime[3]=7,\ 
	\inftime[4]=3,\ 
	\inftime[5]=4.$
\end{example}
\subsection{Computational complexity}\label{sec:compCompl}
%
%
%
\subsubsection{Dynamics (\ref{eq:system})}
in order to compute the disease evolution
using~\eqref{eq:system}, we need $n(n+4)$ operations
for every time step.
Note that the number of time steps needed to capture the entire
evolution depends on the size of parameters $p_{ij}$, $\rectime$ 
and $\tforex$;
hence, with this approach, which is the same described in~\cite{koher2016infections}, we can compute $\inftime$ for all 
nodes only in weakly polynomial time
complexity of $\pazocal{O}(n^2 T)$,
where $T\in\posint$ is the simulation horizon.
\subsubsection{Contagion Graph}
Obtaining a realization
of the Contagion Graph requires at most
a computational complexity equal to $\pazocal{O}(n^2)$. 
In fact, drawing the Contagion Graph
requires to compute (\ref{eq:howcalculatetauij})
for all elements of set $\nodes\times\nodes$.
\newline
By Theorem~\ref{thm:minPath}, for each agent $\aginset$,
$\ki$ can be computed by solving a single-source
shortest path problem from every node in $\exnodes$ to every node in $\nodes$.
%
Since the weights of every Contagion Graph realization
are non-negative, this can be done by applying Dijkstra's
algorithm $|\exnodes|$ times.
Note that the complexity of one run of Dijkstra's algorithm
equals $\pazocal{O}(n^2)$.
Hence, the overall complexity of employing
the Contagion Graph for computing $\{\ki\}_{\aginset}$
is
$$\pazocal{O}(n^2+n^2|\exnodes|)=\pazocal{O}(n^2|\exnodes|),$$
thus strongly polynomial in time, see, e.g.,~\cite{cormen2009introduction}.

\begin{remark}
	Note that standard models usually consider $|\exnodes|=1$, which corresponds to considering one \emph{patient zero}.
    Given that, in general, $T\gg|\exnodes|$,
    using the Contagion Graph
    greatly improves simulation performance.
\end{remark}

\begin{example}[Continues from Example~\ref{ex:time-based}]
	Running the same simulation as \emph{Fig.~\ref{fig:epidemics_pij_02_Rimax_5}}
	with the Contagion Graph
	takes $0.05\mathrm{s}$ on the same machine.
\end{example}
\subsection{Contagion Graph as approximated model}\label{sec:contGraphApprox}
The Contagion Graph can be also a computationally
inexpensive method to assess the
``mean behavior'' of an epidemics
over a network.
We formulate the approach, 
described also in~\cite[section 2.1]{loui1983optimal}, 
only in the case of static networks; 
future work will extend the same strategy to 
the time-varying case.

Consider an arc, say $(j,i)$, of the Contagion Graph.
Its weight is a random variable (moreover, in case
its realization is infinite, the arc is dropped).
We aim at estimating the value of this arc weight
by fitting some parameters.
In fact, given a parameter $\beta\in[0,1]$, 
the estimated arc weight $\estforinf\in\posint\cup\{\infty\}$ 
is the minimum value greater than~$\tforinf$ with 
probability higher than~$\beta$. Formally, $\forall(i,j)\in\arcs_1$,
\eq\label{eq:estforinf}
	\estforinf := \inf \{\tau \in\posint \ |\  \probab{\tau\geq\tforinf} \geq \beta\}.
\qe
Different choices of $\beta$ correspond to different approximations of $\tforinf$.
By (\ref{eq:SIR_distrib}) and (\ref{eq:estforinf}),
%
one can obtain an explicit formula for $\estforinf$, i.e., 
\footnote{
	Relaxing the condition of discrete $\estforinf$, the ceiling function can be removed.}:
\[
	\estforinf = \phi_{\rectime[j]}\left( \ceil{\log_{1-p_{ij}}(1-\beta)} \right).
\]
We build a Contagion Graph with arcs having
weights equal to $\estforinf$ (note that if 
the weight is infinite, the corresponding arc is dropped).
Thus, we can obtain an 
estimation of set $\{\ki \}_{\aginset}$ (depending
on $\beta$), by solving the shortest path problem, as in Section~\ref{su:contagion_equivalent}.
Clearly, these estimates are computed in the same time complexity of a single simulation.
\begin{figure}
	\centering
	\includegraphics[width=.6\columnwidth]{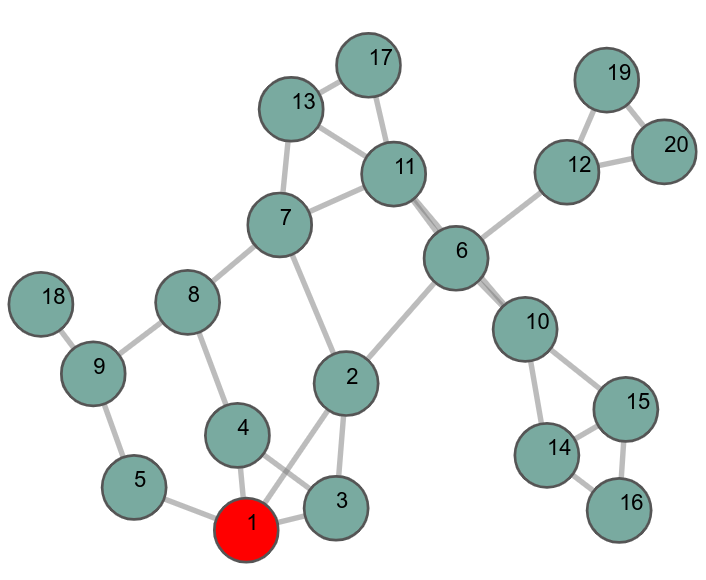}
	\caption{Network topology, with agent~$1$ in red.}
	\label{fig:network_CG}
\end{figure}
\begin{example}\label{ex:CGavg}
	Consider the network in \emph{Fig.~\ref{fig:network_CG}},
	in which $u_1(0)=1$. 
	We consider two different scenarios for the problem:
	(i) $p_{ij}=0.8$ and $\rectime\in[3,20]$,
	(ii) $p_{ij}=0.2$ and $\rectime\in[3,20]$.
	For each scenario, 
	we run $400$ \textit{Montecarlo} simulations
	employing dynamics~(\ref{eq:system}) 
	and one realization of the Contagion Graph following
	the idea in Section~\ref{sec:contGraphApprox},
	with
	$\beta=0.5$.
	We compare results of the Contagion Graph
	with results from the discrete-time simulation, thus showing
	that the Contagion Graph could be used as an approximated model.
	In fact, in \emph{Fig.~\ref{fig:montecarlo_pij08_Ri20}} (scenario (i)),
	one can see that the Contagion Graph represents the
	``mean behavior'' of the randomized simulations.
	By decreasing the infectivity (scenario (ii)), many more nodes will happen
	to be non-infected through simulations. This is the case of 	\emph{Fig.~\ref{fig:montecarlo_pij02_Ri20}}, in which 
	the Contagion Graph captures the fact that some nodes are not expected to get
	infected, and these correspond to the nodes that 
	are found to be non-infected most times in the Montecarlo simulations.
\end{example}

\begin{figure}[h]
	\centering
	\includegraphics[width=\columnwidth]{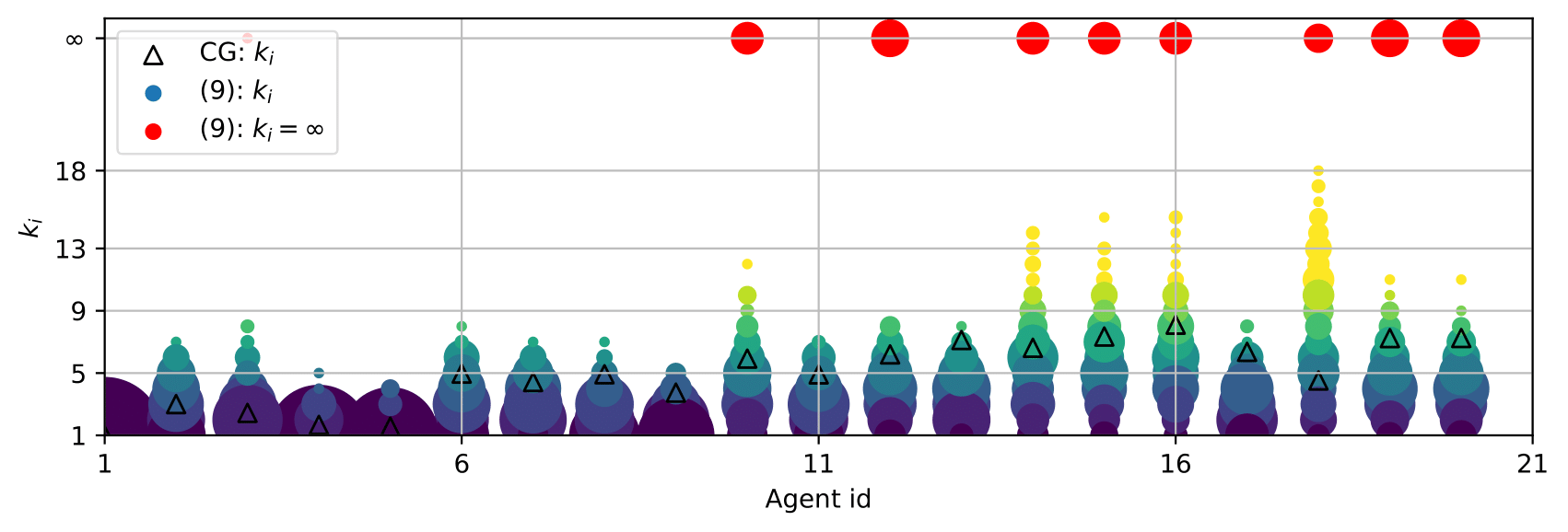}
	\caption{Simulation results for Example~\ref{ex:CGavg}-(i).
	The diameter of each circle is proportional to
	how many times $k_i$ is found with that value in 
	the Montecarlo simulations of dynamics~(\ref{eq:system}).
	Few agents are found to be non-infected (in few simulations).
	The results provided by the Contagion Graph 
	(black triangles) show that this approach provides
	a reliable estimation of the ``mean behavior''.
}
	\label{fig:montecarlo_pij08_Ri20}
\end{figure}
\begin{figure}[h!]
	\centering
	\includegraphics[width=\columnwidth]{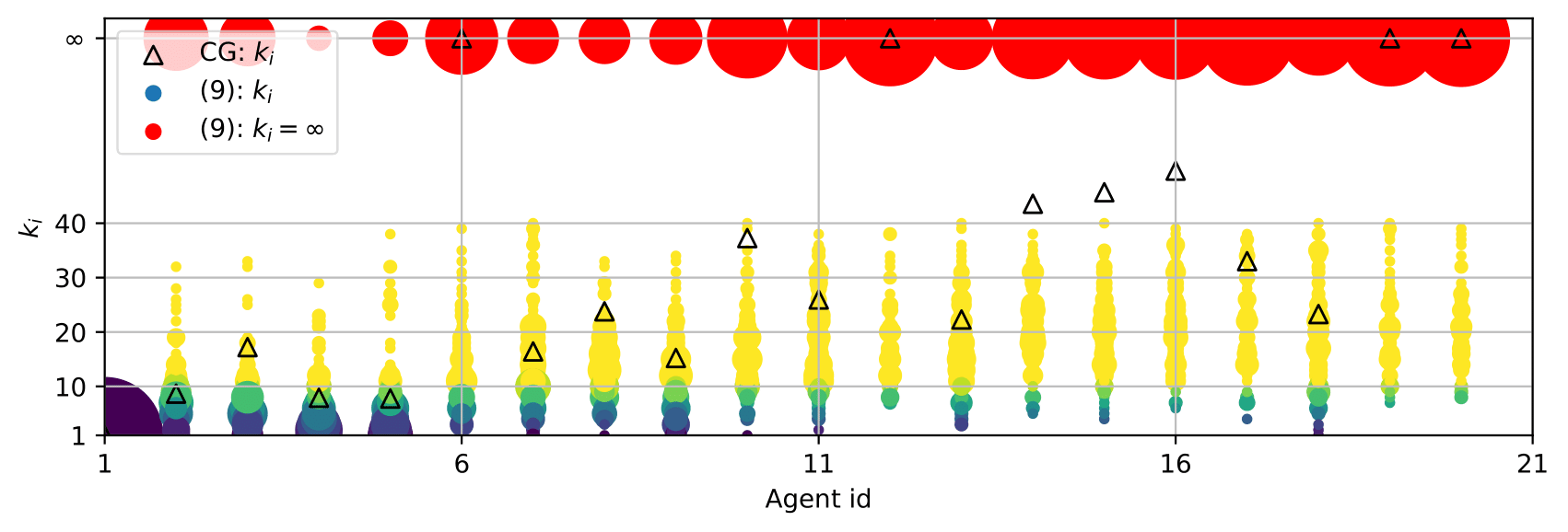}
	\caption{Simulation results for Example~\ref{ex:CGavg}-(ii).
	The decrease in \textit{infectivity}
	results in many agents not being infected through
	the Montecarlo simulations (see the red circles with larger diameter than Example~\ref{ex:CGavg}-(i)).
	This phenomenon is captured by the Contagion Graph.
}
	\label{fig:montecarlo_pij02_Ri20}
\end{figure}

%
%
		
	\section{Conclusion and Future Work} \label{se:conclusion}
This paper has considered an agent-based
model of epidemics over complex networks
and has provided a novel method
for simulating epidemics with lower computational
complexity. The proposed method is based on a graph-based formalization of the problem
and can be also employed for estimating the mean behavior of the epidemic.

Future work will aim at extending the present work to
the case of time-varying contact networks, and at developing
fast control actions based on the \textit{Contagion Graph}
for containing epidemics on large-scale networks.
By using available open data, we aim to employ our method to  investigate COVID-19 scenarios.

	\bibliographystyle{plain} 
	\bibliography{literature}           

\begin{thebibliography}{10}

\bibitem{ahmad2019continuous}
Rehan Ahmad and Kevin~S Xu.
\newblock Continuous-time simulation of epidemic processes on dynamic
  interaction networks.
\newblock In {\em International Conference on Social Computing,
  Behavioral-Cultural Modeling and Prediction and Behavior Representation in
  Modeling and Simulation}, pages 143--152. Springer, 2019.

\bibitem{bouchnita2020hybrid}
Anass Bouchnita and Aissam Jebrane.
\newblock A hybrid multi-scale model of {COVID-19} transmission dynamics to
  assess the potential of non-pharmaceutical interventions.
\newblock {\em Chaos, Solitons \& Fractals}, page 109941, 2020.

\bibitem{calafiore2020modified}
Giuseppe~C Calafiore, Carlo Novara, and Corrado Possieri.
\newblock A modified {SIR} model for the {COVID}-19 contagion in {I}taly.
\newblock {\em arXiv preprint arXiv:2003.14391}, 2020.

\bibitem{cormen2009introduction}
Thomas~H Cormen, Charles~E Leiserson, Ronald~L Rivest, and Clifford Stein.
\newblock {\em Introduction to algorithms}.
\newblock MIT press, 2009.

\bibitem{devroye1986non}
L~Devroye.
\newblock {\em Non-uniform random variate generation}.
\newblock Springer-Verlag, 1986.

\bibitem{duan2015mathematical}
Wei Duan, Zongchen Fan, Peng Zhang, Gang Guo, and Xiaogang Qiu.
\newblock Mathematical and computational approaches to epidemic modeling: a
  comprehensive review.
\newblock {\em Frontiers of Computer Science}, 9(5):806--826, 2015.

\bibitem{giordano2020sidarthe}
Giulia Giordano, Franco Blanchini, Raffaele Bruno, Patrizio Colaneri,
  Alessandro Di~Filippo, Angela Di~Matteo, Marta Colaneri, et~al.
\newblock A {SIDARTHE} model of {COVID-19} epidemic in italy.
\newblock {\em arXiv preprint arXiv:2003.09861}, 2020.

\bibitem{koher2016infections}
Andreas Koher, Hartmut~HK Lentz, Philipp H{\"o}vel, and Igor~M Sokolov.
\newblock Infections on temporal networks - a matrix-based approach.
\newblock {\em PloS one}, 11(4):e0151209, 2016.

\bibitem{loui1983optimal}
Ronald~Prescott Loui.
\newblock Optimal paths in graphs with stochastic or multidimensional weights.
\newblock {\em Communications of the ACM}, 26(9):670--676, 1983.

\bibitem{osthus2017forecasting}
Dave Osthus, Kyle~S Hickmann, Petru{\c{t}}a~C Caragea, Dave Higdon, and Sara~Y
  Del~Valle.
\newblock Forecasting seasonal influenza with a state-space {SIR} model.
\newblock {\em The annals of applied statistics}, 11(1):202, 2017.

\bibitem{pellis2015eight}
Lorenzo Pellis, Frank Ball, Shweta Bansal, Ken Eames, Thomas House, Valerie
  Isham, and Pieter Trapman.
\newblock Eight challenges for network epidemic models.
\newblock {\em Epidemics}, 10:58--62, 2015.

\bibitem{rocha2016individual}
Luis~EC Rocha and Naoki Masuda.
\newblock Individual-based approach to epidemic processes on arbitrary dynamic
  contact networks.
\newblock {\em Scientific reports}, 6:31456, 2016.

\bibitem{teweldemedhin2004agent}
Eyob Teweldemedhin, Tshilidzi Marwala, and Conrad Mueller.
\newblock Agent-based modelling: a case study in {HIV} epidemic.
\newblock In {\em Fourth International Conference on Hybrid Intelligent Systems
  (HIS'04)}, pages 154--159. IEEE, 2004.

\bibitem{vestergaard2015temporal}
Christian~L Vestergaard and Mathieu G{\'e}nois.
\newblock Temporal {G}illespie algorithm: Fast simulation of contagion
  processes on time-varying networks.
\newblock {\em PLoS Comput Biol}, 11(10):e1004579, 2015.

\bibitem{youssef2013mitigation}
Mina Youssef and Caterina Scoglio.
\newblock Mitigation of epidemics in contact networks through optimal contact
  adaptation.
\newblock {\em Mathematical biosciences and engineering: MBE}, 10(4):1227,
  2013.

\end{thebibliography}
\end{document}